\documentclass[a4paper]{article}

\usepackage{amssymb, amsmath, graphicx, hyperref,a4wide, fullpage}
\usepackage[utf8]{inputenc}
\usepackage[noadjust]{cite}
\usepackage{todonotes}
\usepackage{enumitem}

\newcommand{\aparagraph}[1]{\paragraph{#1}}

 \newtheorem{lemma}{Lemma}
 \newtheorem{theorem}{Theorem}

\newcommand{\qed}{\hfill\ensuremath{\Box}\medskip\\\noindent}
\newenvironment{proof}{\noindent\emph{Proof. }}

\newcommand{\level}{\ensuremath\mathrm{level}}

\newcommand{\access}{\ensuremath\mathsf{access}}
\newcommand{\setfinger}{\ensuremath\mathsf{setfinger}}
\newcommand{\movefinger}{\ensuremath\mathsf{movefinger}}
\newcommand{\lce}{\ensuremath\mathsf{lce}}

\newcommand{\str}{\ensuremath{S} }

\newcommand{\slp}{\ensuremath{\mathcal{S}} }

\newcommand{\nca}{\ensuremath{\textsf{nca}}}

\newcommand{\leftc}{\ensuremath{\text{\textit{left}}}}
\newcommand{\rightc}{\ensuremath{\text{\textit{right}}}}

\title{Finger Search in Grammar-Compressed Strings}
\author{Philip Bille \\ \texttt{phbi@dtu.dk} \and Anders Roy Christiansen \\ \texttt{aroy@dtu.dk} \and Patrick Hagge Cording \\ \texttt{phaco@dtu.dk} \and Inge Li Gørtz \\ \texttt{inge@dtu.dk}} 
\date{}
\begin{document}
\maketitle


\begin{abstract}
	Grammar-based compression, where one replaces a long string by a small context-free grammar that generates the string, is a simple and powerful paradigm that captures many popular compression schemes. Given a grammar, the random access problem is to compactly represent the grammar while supporting random access, that is, given a position  in the original uncompressed string report the character at that position. In this paper we study the random access problem with the finger search property, that is, the time for a random access query should depend on the distance between a specified index $f$, called the \emph{finger}, and the query index $i$. We consider both a static variant, where we first place a finger and subsequently access indices near the finger efficiently, and a dynamic variant where also moving the finger such that the time depends on the distance moved is supported. 
	
	Let $n$ be the size the grammar, and let $N$ be the size of the string. For the static variant we give a linear space representation that supports placing the finger in $O(\log N)$ time and subsequently accessing in $O(\log D)$ time, where $D$ is the distance between the finger and the accessed index. For the dynamic variant we give a linear space representation that supports placing the finger in $O(\log N)$ time and accessing and moving the finger in $O(\log D + \log \log N)$ time. Compared to the best linear space solution to random access, we improve a $O(\log N)$ query bound to $O(\log D)$ for the static variant and to $O(\log D + \log \log N)$ for the dynamic variant, while maintaining linear space. As an application of our results we obtain an improved solution to the longest common extension problem in grammar compressed strings. To obtain our results, we introduce several new techniques of independent interest, including a novel van Emde Boas style decomposition of grammars.
\end{abstract}

\section{Introduction}

Grammar-based compression, where one replaces a long string by a small context-free grammar that generates the string, is a simple and powerful paradigm that captures many popular compression schemes including the Lempel-Ziv family~\cite{lz78, lz77, Welch1984}, Sequitur~\cite{nevill1997identifying}, Run-Length Encoding, Re-Pair~\cite{larsson2000off}, and many more~\cite{Shibata-et-al-1999,Gage1994,
	KiefferYang2000,Kiefferetal2000,YangKieffer2000,Apostolico2000off,ApostolicoLonardi1998,ApostolicoLonardi2000a,goto2015lzd}. All of these are or can be transformed into equivalent grammar-based
compression schemes with little expansion~\cite{Rytter2003,Charikaretal2005}. 

Given a grammar $\mathcal{S}$ representing a string $S$, the \emph{random access problem} is to compactly represent $\mathcal{S}$ while supporting fast $\access$ queries, that is, given an index $i$ in $S$ to report $S[i]$. The random access problem is one of the most basic primitives for computation on grammar compressed strings, and solutions to the problem are a key component in a wide range of  algorithms and data structures for grammar compressed strings  \cite{bille2013fingerprints, BLRSSW2014, gagie2012faster, gagie2014lz77, DBLP:journals/jda/GagieGP15, Bille2015, tomohiro2015detecting,tanaka2013computing,tomohiro2015compressed,BCPT2015}. 

In this paper we study the random access problem with the \emph{finger search property}, that is, the time for a random access query should depend on the distance between a specified index $f$, called the \emph{finger}, and the query index $i$. We consider two variants of the problem. The first variant is \emph{static finger search}, where we can place a finger with a $\setfinger$ operation and subsequently access positions near the finger efficiently. The finger can only be moved by a new $\setfinger$ operation, and the time for $\setfinger$ is independent of the distance to the previous position of the finger. The second variant is \emph{dynamic finger search}, where we also support a $\movefinger$ operation that updates the finger such that the update time depends on the distance the finger is moved.

Our main result is efficient solutions to both finger search problems. To state the bounds, let $n$ be the size the grammar $\mathcal{S}$, and let $N$ be the size of the string $S$. For the static finger search problem, we give an $O(n)$ space representation that supports $\setfinger$ in $O(\log N)$ time and $\access$ in $O(\log D)$ time, where $D$ is the distance between the finger and the accessed index. For the dynamic finger search problem, we give an $O(n)$ space representation that supports $\setfinger$ in $O(\log N)$ time and $\movefinger$ and $\access$ in $O(\log D + \log \log N)$ time. The best linear space solution for the random access problem uses $O(\log N)$ time for $\access$. Hence, compared to our result we improve the $O(\log N)$ bound to $O(\log D)$ for the static version and to $O(\log D + \log \log N)$ for the dynamic version, while maintaining linear space. These are the first non-trivial bounds for the finger search problems. 

As an application of our results we also give a new solution to the \emph{longest common extension problem} on grammar compressed strings~\cite{bille2013fingerprints, tomohiro2015detecting, NishimotoIIBT15}. Here, the goal is to compactly represent $\mathcal{S}$ while supporting fast $\lce$ queries, that is, given a pair of indices $i,j$ to compute the length of the longest common prefix of $S[i,N]$ and $S[j,N]$. We give an $O(n)$ space representation that answers queries in $O(\log N + \log^2 \ell)$, where $\ell$ is the length of the longest common prefix. The best $O(n)$ space solution for this problem uses $O(\log N\log \ell)$ time, and hence our new bound is always at least as good and better whenever $\ell = o(N^\varepsilon)$. 

\subsection{Related Work}
We briefly review the related work on the random access problem and finger search. 

\aparagraph{Random Access in Grammar Compressed Strings} 
First note that naively we can store $S$ explicitly using $O(N)$ space and report any character in constant time. Alternatively, we can compute and store the sizes of the strings derived by each grammar symbol in $\mathcal{S}$ and use this to simulate a top-down search on the grammars derivation tree in constant time per node. This leads to an $O(n)$ space representation using $O(h)$ time, where $h$ is the height of the grammar~\cite{gasieniec2005real}. Improved succinct space representation of this solution are also known~\cite{claude2011self}. Bille et al.~\cite{BLRSSW2014} gave a solution using $O(n)$ and $O(\log N)$ time, thus achieving a query time independent of the height of the grammar. Verbin and Yu~\cite{verbin2013data} gave a near matching lower bound by showing that any solution using $O(n \log^{O(1)} N)$ space must use $\Omega(\log^{1- \epsilon} N)$ time. Hence, we cannot hope to obtain significantly faster query times within $O(n)$ space. Finally, Belazzougui et al.~\cite{BCPT2015} very recently showed that with superlinear space slightly faster query times are possible. Specifically, they gave a solution using $O(n\tau\log_\tau N/n)$ space and $O(\log_\tau N)$ time, where $\tau$ is a trade-off parameter. For $\tau = \log^\epsilon N$ this is $O(n \log^\epsilon N)$ space and $O(\log N/\log \log N)$ time. Practical solutions to this problem have been considered in \cite{BelazzouguiBG, navarro2014grammar,gagie2014block}.

The above solutions all generalize to support decompression of an arbitrary substring of length $D$ in time $O(t_{\access} + D)$, where $t_{\access}$ is the time for $\access$ (and even faster for small alphabets~\cite{BCPT2015}). We can extend this to a simple solution to finger search (static and dynamic). The key idea is to implement $\setfinger$ as a random access and $\access$ and $\movefinger$ by decompressing or traversing, respectively, the part of the grammar in-between the two positions. This leads to a solution that uses $O(t_{\access})$ time for $\setfinger$ and $O(D)$ time for $\access$ and $\movefinger$. 

Another closely related problem is the \emph{bookmarking problem}, where a set of positions, called \emph{bookmarks}, are given at preprocessing time and the goal is to support fast substring decompression from any bookmark in constant or near-constant time per decompressed character~\cite{gagie2012faster,cordingbookmarks}. In other words, bookmarking allows us to decompress a substring of length $D$ in time $O(D)$ if the substring crosses a bookmark. Hence, with bookmarking we can improve the $O(t_{\access} + D)$ time solution for substring decompression to $O(D)$ whenever we know the positions of the substrings we want to decompress at preprocessing time. A key component in the current solutions to bookmarking is to trade-off the $\Omega(D)$ time we need to pay to decompress and output the substring. Our goal is to support access without decompressing in $o(D)$ time and hence this idea does not immediately apply to finger search.  

\aparagraph{Finger Search}
Finger search is a classic and well-studied concept in data structures, see e.g.,  \cite{Bentley1976,BMW2003,BLMTT2003,SeidelA96,DR1994,Fleischer1996,GuibasMPR77,Mehlhorn1981,Kosaraju1981,Pugh1990,Sleator1985} and the survey~\cite{Brodal04}. In this setting, the goal is to maintain a dynamic dictionary data structure such that searches have the finger search property. Classic textbook examples of efficient finger search dictionaries include splay trees, skip lists, and level linked trees. Given a comparison based dictionary with $n$ elements, we can support optimal searching in $O(\log n)$ time and finger searching in $O(\log d)$ time, where $d$ is the rank distance between the finger and the query~\cite{Brodal04}. Note the similarity to our compressed results that reduce an $O(\log N)$ bound to $O(\log D)$.  

\subsection{Our results} 
We now formally state our results. Let $S$ be a string of length $N$ compressed into a grammar $\mathcal{S}$ of length $n$. Our goal is to support the following operations on $\mathcal{S}$. 
\begin{description}
	\item[$\quad\access(i)$:] return the character $S[i]$
	\item[$\quad\setfinger(f)$:] set the finger at position $f$ in $S$.  
	\item[$\quad\movefinger(f)$:] move the finger to position $f$ in $S$.
\end{description}
The static finger problem is to support $\access$ and $\setfinger$, and the dynamic finger search problem is to support all three operations. We obtain the following bounds for the finger search problems. 
\begin{theorem}\label{thm:finger}
	Let $\slp$ be a grammar of size $n$ representing a string $S$ of length $N$. Let $f$ be the current position of the finger, and let $D = |f - i|$ for some $i$. Using $O(n)$ space we can support either:
	\begin{enumerate}[label=(\roman*)]
		\item $\setfinger(f)$ in $O(\log N)$ time and $\access(i)$ in $O(\log D)$ time.
		\item $\setfinger(f)$ in $O(\log N)$ time, $\movefinger(i)$ and $\access(i)$ both in $O(\log D + \log\log N)$ time.
	\end{enumerate}	
\end{theorem}
Compared to the previous best linear space solution, we improve the $O(\log N)$ bound to $O(\log D)$ for the static variant and to $O(\log D + \log \log N)$ for the dynamic variant, while maintaining linear space. These are the first non-trivial solutions to the finger search problems. Moreover, the logarithmic bound in terms of $D$ may be viewed as a natural grammar compressed analogue of the classic uncompressed finger search solutions. We note that Theorem~\ref{thm:finger} is straightforward to generalize to multiple fingers. Each additional finger can be set in $O(\log N)$ time, uses $O(\log N)$ additional space, and given any finger $f$, we can support $\access(i)$ in $O(\log D_f)$ time, where $D_f = |f - i|$.

\subsection{Technical Overview}
To obtain Theorem~\ref{thm:finger} we introduce several new techniques of independent interest. First, we consider a variant of the random access problem, which we call the \emph{fringe access problem}. Here, the goal is to support fast access close to the beginning or end (the fringe) of a substring derived by a grammar symbol. We present an $O(n)$ space representation that supports fringe access from any grammar symbol $v$ in time $O(\log D_v + \log \log N)$, where $D_v$ is the distance from the fringe in the string $S(v)$ derived by $v$ to the queried position. The key challenge is designing a data structure for efficient navigation in unbalanced grammars. 

The main component in our solution to this problem is a new recursive decomposition. The decomposition resembles the classic van Emde Boas data structure~\cite{van1976design}, in the sense that we recursively partition the grammar into a hierarchy of depth $O(\log \log N)$ consisting of subgrammars generating strings of lengths $N^{1/2}, N^{1/4}, N^{1/8}, \ldots$. We then show how to implement fringe access via predecessor queries on special paths produced by the decomposition. We cannot afford to explicitly store a predecessor data structure for each special path, however, using a technique due to Bille et al.~\cite{BLRSSW2014}, we can represent all the special paths compactly in a tree and instead implement the predecessor queries as weighted ancestor queries on the tree. This leads to an $O(n)$ space solution with $O(\log D_v + (\log \log N)^2)$ query time. Whenever $D_v \geq 2^{(\log \log N)^2}$ this matches our desired bound of $O(\log D_v + \log \log N)$. To handle the case when $D_v \leq 2^{(\log \log N)^2}$ we use an additional decomposition of the grammar and further reduce the problem to weighted ancestor queries on trees of small weighted height. Finally, we give an efficient solution to weighted ancestor for this specialized case that leads to our final result for fringe access. 

Next, we use our fringe access result to obtain our solution to the static finger search problem. The key idea is to decompose the grammar into heavy paths as done by Bille et al.~\cite{BLRSSW2014}, which has the property that any root-to-leaf path in the directed acyclic graph representing the grammar consists of at most $O(\log N)$ heavy paths. We then use this to compactly represent the finger as a sequence of the heavy paths. To implement $\access$, we binary search the heavy paths in the finger to find an exit point on the finger, which we then use to find an appropriate node to apply our solution to fringe access on. Together with a few additional tricks this gives us Theorem~\ref{thm:finger}(i).

Unfortunately, the above  approach for the static finger search problem does not extend to the dynamic setting. The key issue is that even a tiny local change in the position of the finger can change $\Theta(\log N)$ heavy paths in the representation of the finger, hence requiring at least $\Omega(\log N)$ work to implement $\movefinger$. To avoid this we give a new compact representation of the finger based on both heavy path and the special paths obtained from our van Emde Boas decomposition used in our fringe access data structure. We show how to efficiently maintain this representation during local changes of the finger, ultimately leading to Theorem~\ref{thm:finger}(ii).

\subsection{Longest Common Extensions} 
As application of Theorem~\ref{thm:finger}, we give an improved solution to longest common extension problem in grammar compressed strings. The first solution to this problem is due to Bille et al.~\cite{bille2013fingerprints}. They showed how to extend random access queries to compute Karp-Rabin fingerprints. Combined with an exponential search this leads to a linear space solution to the longest common extension problem using $O(\log N\log \ell)$ time, where $\ell$ is the length of the longest common extension. We note that we can plug in any of the above mentioned random access solution. More recently, Nishimoto~et~al.~\cite{NishimotoIIBT15} used a completely different approach to get $O(\log N + \log \ell\log^* N)$ query time while using superlinear $O(n\log N\log^* N)$ space. We obtain:

\begin{theorem}\label{thm:LCEres}
	Let $\mathcal{S}$ be a grammar of size $n$ representing a string $S$ of length $N$. We can solve the longest common extension problem in $O(\log N + \log^2 \ell)$ time and $O(n)$ space where $\ell$ is the length of the longest common extension.
\end{theorem}

Note that we need to verify the Karp-Rabin fingerprints during preprocessing in order to obtain a worst-case query time. Using the result from Bille et al.~\cite{BLRSSW2014} this gives a randomized expected preprocessing time of $O(N\log N)$. 

Theorem~\ref{thm:LCEres} improves the $O(\log N\log \ell)$ solution to $O(\log N + \log^2 \ell)$. The new bound is always at least as good and asymptotically better whenever $\ell = o(N^\epsilon)$ where $\epsilon$ is a constant. The new result follows by extending Theorem~\ref{thm:finger} to compute Karp-Rabin fingerprints and use these to perform the exponential search from~\cite{bille2013fingerprints}. 


\section{Preliminaries}

\aparagraph{Strings and Trees}
Let $S = S[1, |S|]$ be a string of length $|S|$. Denote by $S[i]$ the character in $S$ at index $i$ and let $S[i, j]$ be the substring of $S$ of length $j - i+1$ from index $i \geq 1$ to $|S| \geq j \geq i$, both indices included.

Given a rooted tree $T$, we denote by $T(v)$ the subtree rooted in a node $v$ and the left and right child of a node $v$ by $\leftc(v)$ and $\rightc(v)$ if the tree is binary. The \textit{nearest common ancestor} $\nca(v,u)$ of two nodes $v$ and $u$ is the deepest node that is an ancestor of both $v$ and $u$. A weighted tree has weights on its edges. A \textit{weighted ancestor} query for node $v$ and weight $d$ returns the highest node $w$ such that the sum of weights on the path from the root to $w$ is at least $d$.

\aparagraph{Grammars and Straight Line Programs}
Grammar-based compression replaces a long string by a small context-free grammar (CFG). We assume without loss of generality that the grammars are in fact \emph{straight-line programs} (SLPs). The lefthand side of a grammar rule in an SLP has exactly one variable, and the forighthand side has either exactly two variables or one terminal symbol. In addition, SLPs are unambigous and acyclic. We view SLPs as a directed acyclic graph (DAG) where each rule correspond to a node with outgoing ordered edges to its variables. Let $\mathcal{S}$ be an SLP. As with trees, we denote the left and right child of an internal node $v$ by $\leftc(v)$ and $\rightc(v)$. The unique string $S(v)$ of length $N_v$ is produced by a depth-first left-to-right traversal of $v$ in $\mathcal{S}$ and consist of the characters on the leafs in the order they are visited. The corresponding parse tree for $v$ is denoted $T(v)$. We will use the following results, that provides efficient random access from any node $v$ in $\mathcal{S}$.
\begin{lemma}[\cite{BLRSSW2014}]\label{lem:randomaccess}
	Let $\str$ be a string of length $N$ compressed into a SLP $\slp$ of size $n$. Given a node $v \in \slp$, we can support random access in $S(v)$ in $O(\log N_v)$ time, and at the same time reporting the sequence of heavy paths and their entry- and exit points in the corresponding depth-first traversal of $\slp(v)$. The number of heavy paths visited is $O(\log N_v)$.
\end{lemma}

\paragraph{Karp-Rabin Fingerprints}

For a prime $p$, $2n^{c+4} < p \leq 4n^{c+4}$ and $x \in [p]$ the Karp-Rabin fingerprint~\cite{rabinkarp},
denoted $\phi(S[i,j])$, of the substring $S[i, j]$ is defined as $\phi(S[i,j]) = \sum_{i \le k \leq j} S[k]x^{k-i}  \bmod p$.
The key property is that for a random choice of $x$, two substrings of $S$ match iff their fingerprints match (whp.), thus allowing us to compare substrings in constant time. We use the following well-known properties of fingerprints.
\begin{lemma}\label{lem:fp}
	The Karp-Rabin fingerprints have the following properties:
	
	\begin{itemize}
		\item[1)] Given $\phi(S[i,j])$, the fingerprint $\phi(S[i,j \pm a])$ for some integer $a$, can be computed in $O(a)$ time.
		\item[2)] Given fingerprints $\phi(S[1,i])$ and $\phi(S[1,j])$, the fingerprint $\phi(S[i,j])$ can be computed in $O(1)$ time.
		\item[3)] Given fingerprints $\phi(S_1)$ and $\phi(S_2)$,  the fingerprint $\phi(S_1 \cdot S_2) = \phi(S_1) \oplus\phi(S_2)$ can be computed in $O(1)$ time.
	\end{itemize}
	
\end{lemma}


\section{Fringe Access}
In this section we consider the \emph{fringe access problem}. Here the goal is to compactly represent the SLP, such that for any node $v$, we can efficiently access locations in the string $S(v)$ close to the start or the end of the substring. The fringe access problem is the key component in our finger search data structures. A straightforward solution to the fringe access problem is to apply a solution to the random access problem. For instance if we apply the random access solution from Bille et al. \cite{BLRSSW2014} stated in Lemma \ref{lem:randomaccess} we immediately obtain a linear space solution with $O(\log N_v)$ access time, i.e., the access time is independent of the distance to the start or the end of the string. This is an immediate consequence of the central grammar decomposition technique of \cite{BLRSSW2014}, and does not extend to solve fringe access efficiently. Our main contribution in this section is a new approach that bypasses this obstacle. We show the following result. 

\begin{lemma}\label{lem:distaccess}
	Let $\mathcal{S}$ be an SLP of size $n$ representing a string of length $N$. Using $O(n)$ space, we can support access to position $i$ of any node $v$, in time $O(\log (\min (i, N_v - i)) + \log\log N)$.
\end{lemma}
The key idea in this result is a van Emde Boas style decomposition of $\mathcal{S}$ combined with a predecessor data structure on selected paths in the decomposition. To achieve linear space we reduce the predecessor queries on these paths to a weighted ancestor query. We first give a data structure with query time $O((\log\log N)^2 + \log (\min (i, N_v - i)))$. We then show how to reduce the query time to $O(\log\log N + \log (\min (i, N_v - i)))$ by reducing the query time for small $i$. To do so we introduce an additional decomposition and give a new data structure that supports fast weighted ancestor queries on trees of small weighted height.

For simplicity and without loss of generality we assume that the access point $i$ is closest to the start of $S(v)$, i.e., the goal is to obtain  $O(\log (i) + \log\log N)$ time. By symmetry we can obtain the corresponding result for access points close to the end of $S(v)$.

\subsection{van Emde Boas Decomposition for Grammars}
We first define the vEB decomposition on the parse tree $T$ and then extend it to the SLP $\mathcal{S}$. In the decomposition we use the ART decompostion by Alstrup et al.~\cite{alstrup1998marked}.

\aparagraph{ART Decomposition}
The ART decomposition introduced by Alstrup et al.~\cite{alstrup1998marked} decomposes a tree into a single top tree
and a number of bottom trees. Each bottom tree is a subtree rooted in a node of minimal depth such that the subtree contains no more than $x$ leaves and the top tree is all nodes not in a bottom tree.
The decomposition has the following key property.

\begin{lemma}[\cite{alstrup1998marked}]\label{lem:artdecomp}
	The ART decomposition with parameter $x$ for a rooted tree T with
	$N$ leaves produces a top tree with at most $\frac{N}
	{x+1}$ leaves.
\end{lemma}

We are now ready to define the \emph{van Emde Boas (vEB) decomposition}.

\subparagraph{The van Emde Boas Decomposition} We define the van Emde Boas Decomposition of a tree $T$ as follows.
The \emph{van Emde Boas (vEB) decomposition} of $T$ is obtained by recursively applying an ART decomposition
: Let $v = root(T)$ and $x = \sqrt{N}$. If $N = O(1)$, stop. Otherwise, construct an ART decomposition of $T(v)$ with parameter $x$. For each bottom tree $T(u)$ recursively construct a vEB decomposition with $v = u$ and $x = \sqrt{x}$.

Define the \emph{level} of a node $v$ in $T$ as $\level(v) = \lfloor \log \log N - \log \log N_v \rfloor$ (this corresponds to the depth of the recursion when $v$ is included in its top tree). 

Note that except for the nodes on the lowest level---which are not in any top tree---all nodes belong to exactly one top tree.  For any node $v \in T$ not in the last level, let $T_{top}(v)$ be the top tree $v$ belongs to. The \emph{leftmost top path} of $v$ is the path from $v$ to the \emph{leftmost leaf} of $T_{top}(v)$. See Figure~\ref{fig:artdecomposition}.

Intuitively, the vEB decomposition of $T$ defines a nested hierarchy of subtrees that decrease by at least the square root of the size at each step. 

 \begin{figure}
 	\centering
 	\includegraphics[width=0.4\textwidth]{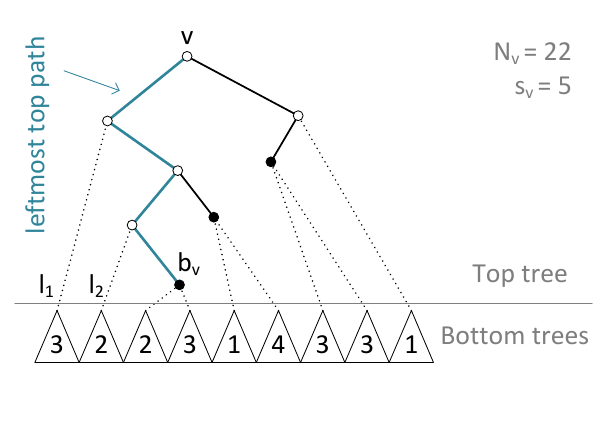}
 	\caption{Example of the ART-decomposition and a leftmost top path. In the top, the nodes forming the top tree are drawn. In the bottom, triangles representing the bottom trees with a number that is the size of the bottom tree. $v$'s leftmost top path is shown as well, and the two trees hanging to the left of this path $l_1$ and $l_2$.}\label{fig:artdecomposition}
 \end{figure}

\aparagraph{The van Emde Boas Decomposition of Grammars}
Our definition of the vEB decomposition of trees can be extended to SLPs as follows. Since the vEB decomposition is based only on the length of the string $N_v$ generated by each node $v$, the definition of the vEB decomposition is also well-defined on SLPs. As in the tree, all nodes belong to at most one top DAG. We can therefore reuse the terminology from the definition for trees on SLPs as well.

To compute the vEB decomposition first determine the level of each node and then remove all edges between nodes on different levels. This can be done in $O(n)$ time.

\subsection{Data Structure}
We first present a data structure that achieves $O((\log \log N)^2 + \log(i))$ time. In the next section we then show how to improve the running time to the desired $O(\log \log (N) + \log(i))$ bound.\\
Our data structure contains the following information for each node $v \in \slp$. Let $l_1, l_2, \ldots, l_k$ be the nodes hanging to the left of $v$'s leftmost top path (excluding nodes hanging from the bottom node).
\begin{itemize}
	\item The length $N_v$ of $S(v)$.
	\item The sum of the sizes of nodes hanging to the left of $v$'s leftmost top path $s_v = |l_1| + |l_2| + \ldots + |l_k|$.
	\item A pointer $b_v$ to the bottom node on $v$'s leftmost top path.
	\item A predecessor data structure over the sequence $1, |l_1|+1, |l_1| + |l_2|+1, \ldots, \sum_{i=1}^{k-1} |l_i|+1$. We will later show how to represent this data structure.
\end{itemize}
In addition we also build the data structure from Lemma~\ref{lem:randomaccess} that given any node $v$ supports random access to $S(v)$ in $O(\log N_v)$ time using $O(n)$ space. 

To perform an access query we proceed as follows. Suppose that we have reached some node $v$ and we want to compute $S(v)[i]$. We consider the following five cases (when multiple cases apply take the first):
\begin{enumerate}
	\item If $N_v = O(1)$. Decompress $S(v)$ and return the $i$'th character.
	\item If $i \leq s_v$. Find the predecessor $p$ of $i$ in $v$'s predecessor structure and let $u$ be the corresponding node. Recursively find $S(u)[i - p]$. 
	\item If $i \leq s_v + N_{\leftc(b_v)}$. Recursively find $S(\leftc(b_v))[i - s_v]$.
	\item If $i \leq s_v + N_{b_v}$. Recursively find $S(\rightc(b_v))[i - s_v - N_{left(b_v)}]$.
	\item In all other cases, perform a random access for $i$ in $\slp(v)$ using Lemma~\ref{lem:randomaccess}. 
\end{enumerate}

To see correctness, first note that case (1) and (5) are correct by definition. Case (2) is correct since when $i \leq s_v$ we know the $i$'th leaf must be in one of the trees hanging to the left of the leftmost top path, and the predecessor query ensures we recurse into the correct one of these bottom trees. In case (3) and (4) we check if the $i$'th leaf is either in the left or right subtree of $b_v$ and if it is, we recurse into the correct one of these.

\aparagraph{Compact Predecessor Data Structures} 

We now describe how to represent the predecessor data structure. Simply storing a predecessor structure in every single node would use $O(n^2)$ space. We can reduce the space to $O(n)$ using ideas similar to the  construction of the "heavy path suffix forest" in  \cite{BLRSSW2014}.

Let $L$ denote the \textit{leftmost top path forest}. The nodes of $L$ are the nodes of $\slp$. A node $u$ is the parent of $v$ in $L$ iff $u$ is a child of $v$ in $\slp$ and $u$ is on $v$'s leftmost top path. Thus, a leftmost top path $v_1, \ldots, v_k$ in $\slp$ is a sequence of ancestors from $v_1$ in $L$. The weight of an edge $(u, v)$ in $L$ is 0 if $u$ is a left child of $v$ in $\slp$ and otherwise $N_{\leftc(v)}$. Several leftmost top paths in $\slp$ can share the same suffix, but the leftmost top path of a node in $\slp$ is uniquely defined and thus $L$ is a forest. A leftmost path ends in a leaf in the top DAG, and therefore $L$ consists of $O(n)$ trees each rooted at a unique leaf of a top dag. A predecessor query on the sequence $1, |l_1|+1, |l_1| + |l_2|+1, \ldots, \sum_{i=1}^{k-1} |l_i|+1$ now corresponds to a weighted ancestor query in $L$. 
We plug in the weighted ancestor data structure from Farach-Colton and Muthukrishnan~\cite{Farach96perfecthashing}, which supports weighted ancestor queries in a forest in $O(\log\log n + \log\log U))$ time with $O(n)$ preprocessing and space, where $U$ is the maximum weight of a root-to-leaf path and $n$ the number of leaves. We have $U = N$ and hence the time for queries becomes $ O(\log \log N)$.

\aparagraph{Space and Preprocessing Time}
For each node in $\slp$ we store a constant number of values, which takes $O(n)$ space. Both the predecessor data structure and the data structure for supporting random access from Lemma~\ref{lem:randomaccess} take $O(n)$ space, so the overall space usage is $O(n)$. The vEB decomposition can be computed in $O(n)$ time. The leftmost top paths and the information saved in each node can be computed in linear time. The predecessor data structure uses linear preprocessing time, and thus the total preprocessing time is $O(n)$.

\aparagraph{Query Time}
Consider each case of the recursion. The time for case (1), (3) and (4) is trivially $O(1)$. Case (2) is $O(\log \log N)$ since we perform exactly one predececssor query in the predecessor data structure. 

In case (5) we make a random access query in a node of size $N_v$. From Lemma~\ref{lem:randomaccess} we have that  the query time is $O(\log N_v)$. We know $\level(v) = \level(b_v)$ since they are on the same leftmost top path. From the definition of the level it follows for any pair of nodes $u$ and $w$ with the same level that $N_u \geq \sqrt{N_w}$ and thus  $N_{b_v} \geq \sqrt{N_v}$. From the conditions we have $i > s_v + N_{b_v} \geq N_{b_v} \geq \sqrt{N_v}$. Since $\sqrt{N_v} < i \Leftrightarrow \log N_v < 2\log i$ we have $\log N_v = O(\log i)$ and thus the running time for  case (5) is $O(\log N_v) = O(\log i)$.

Case (1) and (5) terminate the algorithm and can thus not happen more than once. Case (2), (3) and (4) are repeated at most $O(\log \log N)$ times since the level of the node we recurse on increments by at least one in each recursive call, and the level of a node is at most $O(\log\log N)$. The overall running time is therefore $O((\log\log N)^2 + \log i)$.

In summary, we have the following result.

\begin{lemma}\label{lem:distaccess0}
	Let $\mathcal{S}$ be an SLP of size $n$ representing a string of length $N$. Using $O(n)$ space, we can support access to position $i$ of any node $v$, in time $O(\log i + (\log\log N)^2)$.
\end{lemma}

\subsection{Improving the Query Time for Small Indices}
The above algorithm obtains the running time $O(\log i)$ for $i \geq 2^{(\log\log N)^2}$. We will now improve the running time to $O(\log\log N + \log i)$ by improving the running time in the case when $i < 2^{(\log\log N)^2}$.

In addition to the data structure from above, we add another copy of the data structure with a few changes. When answering a query, we first check if $i \geq 2^{(\log\log N)^2}$. If $i \geq 2^{(\log\log N)^2}$  we use the original data structure, otherwise we use the new copy.

The new copy of the data structure is implemented as follows. In the first level of the ART-decomposition let $x = 2^{(\log\log N)^2}$ instead of $\sqrt N$. For the rest of the levels use $\sqrt{x}$ as before. Furthermore, we split the resulting new leftmost top path forest $L$ into two disjoint parts: $L_1$ consisting of all nodes with level 1 and $L_{\geq2}$ consisting of all nodes with level at least 2. For $L_1$ we use the weighted ancestor data structure by Farach-Colton and Muthukrishnan~\cite{Farach96perfecthashing} as in the previous section using $O(\log\log n + \log\log N) = O(\log \log N)$ time. However, if we apply this solution for $L_{\geq 2}$ we end up with a query time of $O(\log \log n + \log\log x)$, which does not lead to an improved solution. Instead, we present a new data structure that supports queries in $O(\log \log x)$ time. 

\begin{lemma}\label{lem:wa}
	Given a tree $T$ with $n$ leaves where the sum of edge weights on any root-to-leaf path is at most $x$ and the height is at most $x$, we can support weighted ancestor queries in $O(\log\log x)$ time using $O(n)$ space and preprocessing time. 
\end{lemma}

\begin{proof}
	Create an ART-decomposition 
	of $T$ with parameter $x$. 
	For each bottom tree in the decomposition construct the weighted ancestor structure from \cite{Farach96perfecthashing}. For the top tree, construct a predecessor structure over the accumulated edge weights for each root-to-leaf path. 
	
	To perform a weighted ancestor query on a node in a bottom tree, we first perform a weighted ancestor query using the data structure for the bottom tree. In case we end up in the root of the bottom tree, we continue with a predecessor search in the top tree from the leaf corresponding to the bottom tree. 
	
	The total space for bottom trees is $O(n)$. Since the top tree has $O(n/x)$ leaves and height at most $x$, the total space for all predecessor data structures on root-to-leaf paths in the top tree is $O(n/x \cdot x)= O(n)$. Hence, the total space is $O(n)$. 
	
	A predecessor query in the top tree takes $O(\log\log x)$ time. The number of nodes in each bottom tree is at most $x^2$ since it has at most $x$ leaves and height $x$ and the maximum weight of a root-to-leaf path is $x$ giving weighted ancestor queries in $O(\log\log x^2 + \log\log x) = O(\log\log x)$ time. Hence, the total query time is $O(\log\log x)$.\qed
\end{proof}

We reduce the query time for queries with $i < 2^{(\log\log N)^2}$ using the new data structure. 
The level of any node in the new structure is at most $O(1 + \log \log {2^{(\log\log N)^2}}) = O(\log \log \log N)$. A weighted ancestor query in $L_1$ takes time $O(\log\log N)$. For weighted ancestor queries in $L_{\geq 2}$, we know any node $v$ has height at most $2^{(\log\log N)^2}$ and on any root-to-leaf path the sum of the weights is at most $2^{(\log\log N)^2}$. Hence, by  Lemma~\ref{lem:wa} we support queries in  $O(\log \log {2^{(\log\log N)^2}}) = O(\log \log \log N)$ time for nodes in $L_{\geq 2}$.

We make at most one weighted ancestor query in $L_1$, the remaining ones are made in $L_{\geq 2}$, and thus the overall running time is $O(\log\log N + (\log\log\log N)^2 + \log i) = O(\log\log N + \log i)$.

In summary, this completes the proof of Lemma~\ref{lem:distaccess}.


\section{Static Finger Search}\label{sec:sfs}
We now show how to apply our solution to the fringe access to a obtain a simple data structure for the static finger search problem. This solution will be the starting point for solving the dynamic case in the next section, and we will use it as a key component in our result for longest common extension problem.

Similar to the fringe search problem we assume without loss of generality that the access point $i$ is to the right of the finger. 

 \begin{figure}
 	\centering
 	\includegraphics[width=0.44\textwidth]{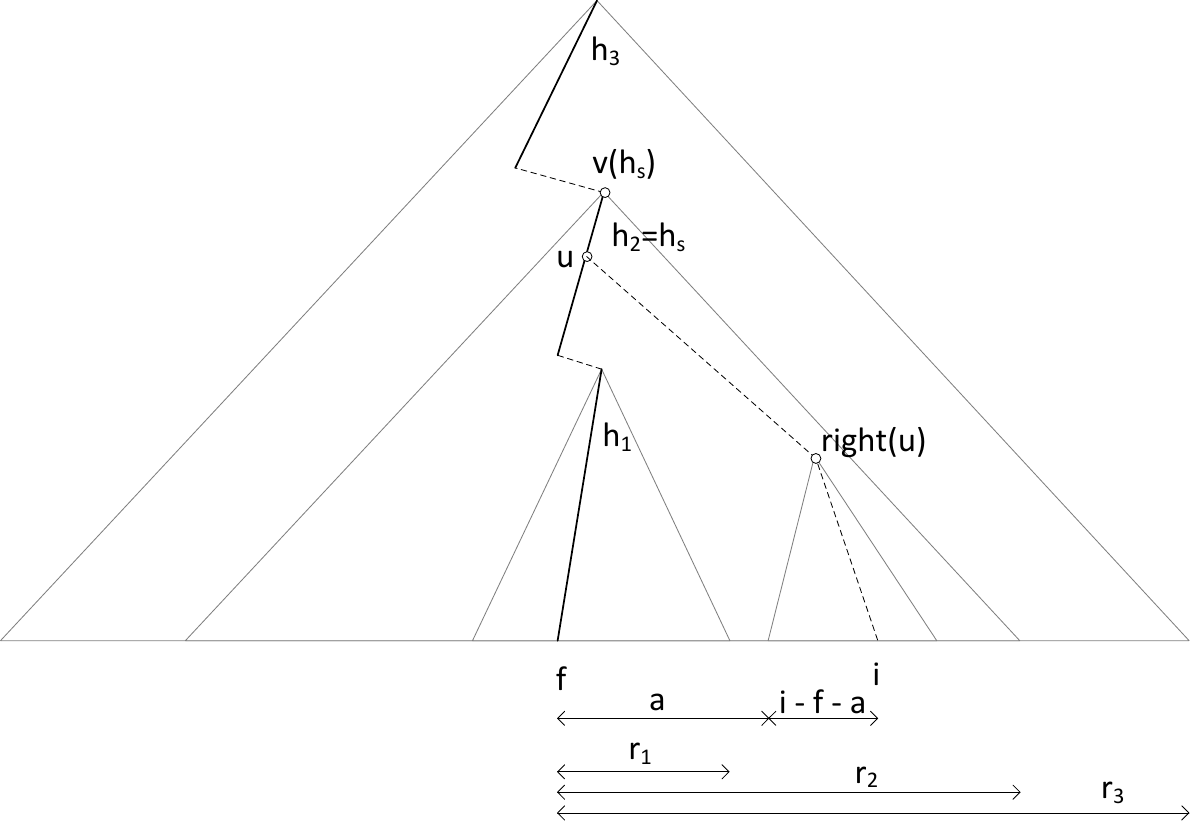}
 	\caption{Illustration of the data structure for a finger pointing at $f$ and an access query at location $i$. $h_1, h_2, h_3$ are the heavy paths visited when finding the finger. $u$ corresponds to $NCA(v_f, v_i)$ in the parse tree and $h_s$ is the heavy path on which $u$ lies, which we use to find $u$. $a$ is a value calculated during the access query. }\label{fig:finger}
 \end{figure}
 
\aparagraph{Data Structure}
We store the random access data structure from \cite{BLRSSW2014} used in Lemma~\ref{lem:randomaccess} and the fringe search data structures from above. Also from \cite{BLRSSW2014} we store the data structure that for any heavy path $h$ starting in a node $v$ and an index $i$ of a leaf in $T(v)$ gives the exit-node from $h$ when searching for $i$ in $O(\log \log N)$ time and uses $O(n)$ space.

To represent a finger the key idea is store a compact data structure for the corresponding root-to-leaf path in the grammar that allows us to navigate it efficiently. Specifically, let $f$ be the position of the current finger and let $p = v_1 \dots v_k$ denote the path in $\slp$ from the root to $v_f$ ($v_1 = root$ and $v_k = v_f$). Decompose $p$ into the $O(\log N)$ heavy paths it intersects, and call these $h_j = v_1 \dots v_{i_1}, h_{j-1} = v_{i_1+1} \dots v_{i_2}, \cdots, h_1 = v_{i_{j-1}+1} \dots v_k$. Let $v(h_i)$ be the topmost node on $h_i$ ($v(h_j) = v_1, v(h_{j-1}) = v_{i_1}, \dots$). Let $l_j$ be the index of $f$ in $\slp(v(h_j))$ and $r_j = N_{v(h_j)} - l_j$. For the finger we store:

\begin{enumerate}
	\item The sequence $r_1, r_2, \dots, r_j$ (note $r_1 \leq r_2 \leq \dots \leq r_j$).
	\item The sequence $v(h_1), v(h_2), \dots, v(h_j)$.
	\item The string $F_T = S[f+1, f+\log N]$.
\end{enumerate}

\aparagraph{Analysis}
The random access and fringe search data structures both require $O(n)$ space. Each of the 3 bullets above require $O(\log N)$ space and thus the finger takes up $O(\log N)$ space. The total space usage is $O(n)$.

\aparagraph{Setfinger} 
We implement $\setfinger(f)$ as follows. First, we apply Lemma~\ref{lem:randomaccess} to make random access to position $f$. This gives us the sequence of visited heavy paths which exactly corresponds to $h_j$, $h_{j-1}, \dots, h_1$ including the corresponding $l_i$ values from which we can calculate the $r_i$ values. So we update the $r_i$ sequence accordingly.
Finally, decompress and save the string $F_T = S[f+1, f+\log N]$.

The random access to position $f$ takes $O(\log N)$ time. In addition to this we perform a constant number of operations for each heavy path $h_i$, which in total takes $O(\log N)$ time. Decompressing a string of $\log N$ characters can be done in $O(\log N)$ time (using \cite{BLRSSW2014}). In total, we use $O(\log N)$ time.

\aparagraph{Access}
To perform $\access(i)$ ($i > f$), there are two cases. If $D = i - f \leq \log N$ we simply return the stored character $F_T[D]$ in constant time. Otherwise, we compute the node $u = \nca(v_f, v_i)$ in the parse tree $T$ as follows. First find the index $s$ of the successor to $D$ in the $r_i$ sequence using binary search. Now we know that $u$ is on the heavy path $h_s$. Find the exit-nodes from $h_s$ when searching for respectively $i$ and $f$ using the data structure from \cite{BLRSSW2014} - the topmost of these two is $u$. See Fig.~\ref{fig:finger}. Finally, we compute $a$ as the index of $f$ in $T(\leftc(u))$ from the right and use the data structure for fringe search from Lemma~\ref{lem:distaccess} to compute $S(\rightc(u))[i - f - a]$. 

For $D \leq \log N$, the operation takes constant time. For $D > \log N$, the binary search over a sequence of $O(\log N)$ elements takes $O(\log \log N)$ time, finding the exit-nodes takes $O(\log\log N)$ time, and the fringe search takes $O(\log(i - f - a)) = O(\log D)$ time. Hence, in total $O(\log\log N + \log D) = O(\log D)$ time.

\medskip

This completes the proof of Theorem~\ref{thm:finger}(i). 

\section{Dynamic Finger Search}
In this section we show how to extend the solution from Section \ref{sec:sfs} to handle dynamic finger search. The target is to support the $\movefinger$ operation that will move the current finger, where the time it takes is dependent on how far the finger is moved. Obviously, it should be faster than simply using the $\setfinger$ operation. The key difference from the static finger is a new decomposition of a root-to-leaf path into paths. The new decomposition is based on a combination of heavy paths and leftmost top paths, which we will show first. Then we show how to change the data structure to use this decomposition, and how to modify the operations accordingly. Finally, we consider how to generalize the solution to work when $\movefinger$/$\access$ might both be to the left and right of the current finger, which for this solution is not trivially just by symmetry.

Before we start, let us see why the data structure for the static finger cannot directly be used for dynamic finger. Suppose we have a finger pointing at $f$ described by $\Theta(\log N)$ heavy paths. It might be the case that after a $\movefinger(f + 1)$ operation, it is $\Theta(\log N)$ completely different heavy paths that describes the finger. In this case we must do $\Theta(\log N)$ work to keep our finger data structure updated. This can for instance happen when the current finger is pointing at the right-most leaf in the left subtree of the root.

Furthermore, in the solution to the static problem, we store the substring $S[f+1, f+\log N]$ decompressed in our data structure. If we perform a $\movefinger(f + \log N)$ operation nothing of this substring can be reused. To decompress $\log N$ characters takes $\Omega(\log N)$ time, thus we cannot do this in the $\movefinger$ operation and still get something faster than $\Theta(\log N)$.

\subsection{Left Heavy Path Decomposition of a Path}
Let $p = v_1 \dots v_k$ be a root-to-leaf path in $\slp$. A subpath $p_i = v_a \dots v_b$ of $p$ is a \emph{maximal heavy subpath} if $v_a \dots v_b$ is part of a heavy path and $v_{b+1}$ is not on the same heavy path. Similarly, a subpath $p_i = v_a \dots v_b$ of $p$ is a \emph{maximal leftmost top subpath} if $v_a \dots v_b$ is part of a leftmost top path and $level(v_b) \neq level(v_{b+1})$.

A \emph{left heavy path decomposition} is a decomposition of a root-to-leaf path $p$ into an arbitrary sequence $p_1 \dots p_j$ of maximal heavy subpaths, maximal leftmost top subpaths and (non-maximal) leftmost top subpaths immediately followed by maximal heavy subpaths. 

Define $v(p_i)$ as the topmost node on the subpath $p_i$. Let $l_j$ be the index of the finger $f$ in $\slp(v(p_j))$ and $r_j = N_{v(p_j)} - l_j$. Let $t(p_i)$ be the type of $p_i$; either heavy subpath ($HP$) or leftmost top subpath ($LTP$).

A left heavy path decomposition of a root-to-leaf path $p$ is not unique. The heavy path decomposition of $p$ is always a valid left heavy path decomposition as well. The visited heavy paths and leftmost top paths during fringe search are always maximal and thus is always a valid left heavy path decomposition.

\begin{lemma}\label{lem:leftheavypaths}
	The number of paths in a left heavy path decomposition is $O(\log N)$.
\end{lemma}

\begin{proof}
There are at most $O(\log N)$ heavy paths that intersects with a root-to-leaf path (Lemma~\ref{lem:randomaccess}). Each of these can at most be used once because of the maximality. So there can at most be $O(\log N)$ maximal heavy paths. Each time there is a maximal leftmost top path, the level of the following node on $p$ increases. This can happen at most $O(\log\log N)$ times. Each non-maximal leftmost top path is followed by a maximal heavy path, and since there are only $O(\log N)$ of these, this can happen at most $O(\log N)$ times. Therefore the sequence of paths has length $O(\log N + \log \log N + \log N) = O(\log N)$.
\end{proof}

\subsection{Data Structure}

We use the data structures from \cite{BLRSSW2014} as in the static variant and the fringe access data structure with an extension. In the fringe access data structure there is a predecessor data structure for all the nodes hanging to the left of a leftmost top path. To support $\access$ and $\movefinger$ we need to find a node hanging to the left or right of a leftmost top path. We can do this by storing an identical predecessor structure for the accumulated sizes of the nodes hanging to the right of each leftmost top path. Again, the space usage for this predecessor structure can be reduced to $O(n)$ by turning it into a weighted ancestor problem.

To represent a finger the idea is again to have a compact data structure representing the root-to-leaf path corresponding to the finger. This time we will base it on a left heavy path decomposition instead of a heavy path decomposition. Let $f$ be the current position of the finger. For the root-to-leaf path to $v_f$ we maintain a left heavy path decomposition, and store the following for a finger:

\begin{enumerate}
	\item The sequence $r_1, r_2, \dots, r_j$ ($r_1 \leq r_2 \leq \dots \leq r_j$) on a stack with the last element on top.
	\item The sequence $v(p_1), v(p_2), \dots, v(p_j)$ on a stack with the last element on top.
	\item The sequence $t(p_1), t(p_2), \dots, t(p_j)$ on a stack with the last element on top.
\end{enumerate}

\aparagraph{Analysis} The fringe access data structure takes up $O(n)$ space. For each path in the left heavy path decomposition we use constant space. Using Lemma \ref{lem:leftheavypaths} we have the space usage of this is $O(\log N) = O(n)$.

\aparagraph{Setfinger} 
Use fringe access (Lemma~\ref{lem:distaccess}) to access position $f$. This gives us a sequence of leftmost top paths and heavy paths visited during the fringe access which is a valid left heavy path decomposition. Calculate $r_i$ for each of these and store the three sequences of $r_i$, $v(p_i)$ and $t(p_i)$ on stacks.

The fringe access takes $O(\log f + \log\log N)$ time. The number of subpaths visited during the fringe access cannot be more than $O(\log f + \log\log N)$ and we only perform constant extra work for each of these.

\aparagraph{Access} To implement $\access(i)$ for $i > f$ we have to find $u = \nca(v_i, v_f)$ in the $T$. Find the index $s$ of the successor to $D = i - f$ in $r_1, r_2, \ldots, r_j$ using binary search. We know $\nca(v_i, v_f)$ lies on $p_s$, and $v_i$ is in a subtree that hangs of $p_s$. The exit-nodes from $p_s$ to $v_f$ and $v_i$ are now found - the topmost of these two is $\nca(v_i, v_f)$. If $t(p_s) = HP$ then we can use the same data structure as in the static case, otherwise we perform the predecessor query on the extra predecessor data structure for the nodes hanging of the leftmost top path. Finally, we compute $a$ as the index of $f$ in $S(\leftc(u))$ from the right and use the data structure for fringe access from Lemma~\ref{lem:distaccess} to compute $S(\rightc(u))[i - f - a]$. 

The binary search on $r_1, r_2, \dots, r_j$ takes $O(\log \log N)$ time. Finding the exit-nodes from $p_s$ takes $O(\log \log N)$ in either case. Finally the fringe access takes $O(\log(i - f - a) + \log \log N) = O(\log D + \log \log N)$. Overall it takes $O(\log D + \log \log N)$.

Note the extra $O(\log\log N)$ time usage because we have not decompressed the first $\log N$ characters following the finger.

\aparagraph{Movefinger}
To move the finger we combine the $\access$ and $\setfinger$ operations. Find the index $s$ of the successor to $D = i - f$ in $r_1, r_2, \ldots, r_j$ using binary search. Now we know $u = \nca(v_i, v_f)$ must lie on $p_s$. Find $u$ in the same way as when performing access. From all of the stacks pop all elements above index $s$. Compute $a$ as the index of $f$ in $S(\leftc(u))$ from the right. The finger should be moved to index $i - f - a$ in $\rightc(u)$. First look at the heavy path $\rightc(u)$ lies on and find the proper exit-node $w$ using the data structure from \cite{BLRSSW2014}. Then continue with fringe searh from the proper child of $w$. This gives a heavy path followed by a sequence of maximal leftmost top paths and heavy paths needed to reach $v_i$ from $\rightc(u)$, push the $r_j$, $v(p_j)$, and $t(p_j)$ values for these on top of the respective stacks.

We now verify the sequence of paths we maintain is still a valid left heavy path decomposition. Since fringe search gives a sequence of paths that is a valid left heavy path decomposition, the only problem might be $p_s$ is no longer maximal. If $p_s$ is a heavy path it will still be maximal, but if $p_s$ is a leftmost top path then $level(u)$ and $level(\rightc(u))$ might be equal. But this possibly non-maximal leftmost top path is always followed by a heavy path. Thus the overall sequence of paths remains a left heavy path decomposition.

The successor query in $r_1, r_2, \dots, r_j$ takes $O(\log\log N)$ time. Finding $u$ on $p_i$ takes $O(\log\log N)$ time, and so does finding the exit-node on the following heavy path. Popping a number of elements from the top of the stacks can be done in $O(1)$ time. Finally the fringe access takes $O(\log(i - f - a) + \log \log N) = O(\log D + \log \log n)$ including pushing the right elements on the stacks. Overall the running time is therefore $O(\log D + \log \log n)$.

\subsection{Moving/Access to the Left of the Finger} In the above we have assumed $i > f$, we will now show how this assumption can be removed. It is easy to see we can mirror all data structures and we will have a solution that works for $i < f$ instead. Unfortunately, we cannot just use a copy of each independently, since one of them only supports moving the finger to the left and the other only supports moving to the right. We would like to support moving the finger left and right arbitrarily. This was not a problem with the static finger since we could just make $\setfinger$ in both the mirrored and non-mirrored data structures in $O(\log N)$ time.

Instead we extend our finger data structure. First we extend the left heavy path decomposition to a \emph{left right heavy path decomposition} by adding another type of paths to it, namely \textit{rightmost top paths} (the mirrorred version of leftmost top paths). Thus a \emph{left right heavy path decomposition} is a decomposition of a root-to-leaf path $p$ into an arbitrary sequence $p_1 \dots p_j$ of maximal heavy subpaths, maximal leftmost/rightmost top subpaths and (non-maximal) leftmost/rightmost top subpaths immediately followed by maximal heavy subpaths.  Now $t(p_i) = HP | LTP | RTP$. Furthermore, we save the sequence $l_1, l_2, \dots, l_j$ ($l_j$ being the left index of $f$ in $T(v(p_i))$) on a stack like the $r_1, r_2, \dots, r_j$ values, etc.

When we do $\access$ and $\movefinger$ where $i < f$, the subpath $p_s$ where $nca(v_f, v_i)$ lies can be found by binary search on the $l_j$ values instead of the $r_j$ values. Note the $l_j$ values are sorted on the stack, just like the $r_j$ values. The following heavy path lookup/fringe access should now be performed on $\leftc(u)$ instead of $\rightc(u)$. The remaining operations can just be performed in the same way as before.


\section{Finger Search with Fingerprints and Longest Common Extensions}\label{sec:flce}

We show how to extend our finger search data structure from Theorem~\ref{thm:finger}(i) to support computing fingerprints and then apply the result to compute longest common extensions. First, we will show how to return a fingerprint for $S(v)[1, i]$ when performing access on the fringe of $v$.

\subsection{Fast Fingerprints on the Fringe}
To do this, we need to store some additional data for each node $v \in \slp$. We store the fingerprint $\phi(S(v))$ and the concatenation of the fingerprints of the nodes hanging to the left of the leftmost top path $p_v = \phi(S(l_1)) \oplus \phi(S(l_2)) \oplus \ldots \oplus \phi(S(l_k))$. We also need the following lemma:

\begin{lemma}[\cite{bille2013fingerprints}]\label{lem:randomfingerprint}
	Let $\str$ be a string of length $N$ compressed into a SLP $\slp$ of size $n$. Given a node $v \in \slp$, we can find the fingerprint $\phi(S(v)[1,i])$ where $1 \leq i \leq N_v$ in $O(\log N_v)$ time.
\end{lemma} 

Suppose we are in a node $v$ and we want to calculate the fingerprint $\phi(S(v)[1, i])$. We perform an access query as before, but also maintain a fingerprint $p$, initially $p = \phi(\epsilon)$, computed thus far.  We follow the same five cases as before, but add the following to update $p$:

\begin{enumerate}
	\item From the decompressed $\str(v)$, calculate the fingerprint for $\str(v)[1, i]$, now update $p = p \oplus \phi(\str(v)[1, i])$.
	\item $p = p \oplus (\phi(p_v) \ominus_s \phi(p_u))$.
	\item $p = p \oplus \phi(p_v)$.
	\item $p = p \oplus \phi(p_v) \oplus \phi(S(\leftc(b_v)))$.
	\item Use Lemma~\ref{lem:randomfingerprint} to find the fingerprint for $S(v)[1, i]$ and then update with $p = p \oplus \phi(S(v)[1, i])$.
\end{enumerate}

These extra operations do not change the running time of the algorithm, so we can now find the fingerprint $\phi(S(v)[1, i])$ in time $O(\log\log N + \log (\min (i, N_v - i)))$.

\subsection{Finger Search with Fingerprints}
Next we show how to do finger search while computing fingerprints between the finger $f$ and the access point~$i$. 

When we perform $\setfinger(f)$ we use the algorithm from \cite{bille2013fingerprints} to compute fingerprints during the search of $\mathcal{S}$ from the root to $f$. This allows us to subsequently compute for any heavy path $h_j$ on the root to position $f$ the fingerprint $p(h_j)$ of the concatenation of the strings generated by the subtrees hanging to the left of $h_j$. In addition, we explicitly compute and store the fingerprints $\phi(S[1, f+1]), \phi(S[1, f+2]), \ldots, \phi(S[1, f+\log N + 1])$. In total, this takes $O(\log N)$ time. 

Suppose that we have now performed a $\setfinger(f)$ operation. To implement $\access(i)$, $i > f$, there are two cases. If $D = i - f \leq \log N$ we return the appropriate precomputed fingerprint. Otherwise, we compute the node $u = \nca(v_f, v_i)$ in the parse tree $T$ as before. Let $h$ be the heavy path containing $u$. Using the data structure from \cite{bille2013fingerprints} we compute the fingerprint $p_l$ of the nodes hanging to the left of $h$  above $u$ in constant time. The fingerprint is now obtained as $\phi(S[1, i]) = p_{h_j} \oplus p_l \oplus \phi(S(\rightc(u))[1, (i - f) - a])$, where the latter is found using fringe access with fingerprints in $\rightc(u)$. None of these additions change the asymptotic complexities of Theorem~\ref{thm:finger}(i). Note that with the fingerprint construction in~\cite{bille2013fingerprints} we can guarantee that all fingerprints are collision-free. 

\subsection{Longest Common Extensions}
Using the fingerprints it is now straightforward to implement $\lce$ queries as in~\cite{bille2013fingerprints}. Given a $\lce(i, j)$ query, first set fingers at positions $i$ and $j$. This allows us to get fingerprints of the form $\phi(S[i, i+a])$ or $\phi(S[j, j+a])$ efficiently. Then, we find the largest value $\ell$ such that $\phi(S[i, i+\ell]) = \phi(S[j, j+\ell])$ using a standard exponential search. Setting the two finger uses $O(\log N)$ time and by Theorem~\ref{thm:finger}(i) the at most $O(\log \ell)$ searches in the exponential search take at most $O(\log \ell)$ time. Hence, in total we use $O(\log N + \log^2 \ell)$ time, as desired. This completes the proof of Theorem~\ref{thm:LCEres}. 

\bibliographystyle{abbrv}
\bibliography{references}


\end{document}